\documentclass[runningheads]{llncs}
\usepackage[T1]{fontenc}
\usepackage{graphicx}

\newcommand{\PFP}{\ensuremath{\mathrm{PFP}}}
\newcommand{\rev}{\ensuremath{\mathrm{rev}}}

\begin{document}

\title{Parse indexing for discarding\\short pseudo-MEMs safely}
\author{Travis Gagie\inst{1}\orcidID{0000-0003-3689-327X}}
\authorrunning{T. Gagie}
\institute{Faculty of Computer Science, Dalhousie University, Halifax, Canada \email{travis.gagie@gmail.com}}
\maketitle

\begin{abstract}
\noindent
Brown et al.\ (2025) described a pre-processing step, called $k$-mer based breaking (KeBaB), that speeds up searching for long maximal exact matches (MEMs) between a pattern $P$ and an indexed repetitive text $T$.  KeBaB produces a set of substrings of $P$ called pseudo-MEMs that often have total length much less than $|P|$ but are still guaranteed to contain all the MEMs of length at least a fixed parameter $k$.  Brown et al.\ found that KeBaB can be particularly effective when we discard all but the longest pseudo-MEMs --- but then we risk also discarding the longest MEMs!  In this paper we show how we can use parse indexing to generate pseudo-MEMs together with lower bounds on the lengths of the longest MEMs they must contain, allowing us to discard short pseudo-MEMs safely.

\end{abstract}

\section{Introduction}
\label{sec:introduction}

A {\em maximal exact match} (MEM) of a pattern $P [1..m]$ with respect to a text $T [1..n]$ is a substring $P [i..j]$ of $P$ that occurs in $T$ but such that $i = 1$ or $P [i - 1..j]$ does not occur in $T$, and $j = m$ or $P [i..j + 1]$ does not occur in $T$.  Finding MEMs between patterns and indexed texts is a standard task in bioinformatics (see, e.g.,~\cite{Li13,Li24}) for which the most popular solution is Li's~\cite{Li12} forward-backward algorithm, which usually uses one FM-index~\cite{FM05} for $T$ and another for its reverse $T^\rev$ and takes a number of backward steps in those indexes proportional to the total length of the MEMs.  The FM-indexes can be based on the Burrows-Wheeler Transforms (BWT) of $T$ and $T^\rev$ or run-length compressed BWTs of them, and need not store a suffix-array sample.  If $T$ is over a large alphabet then it is better in practice to use compressed suffix arrays (CSA)~\cite{GV05} or run-length compressed CSAs (RLCSAs)~\cite{MNSV10} rather than FM-indexes, since they have better locality; see~\cite{BGMNS25,CGN26} for recent discussions.

Gagie~\cite{Gag24} proposed a modification of forward-backward, now called Boyer-Moore-Li (BML) because of its similarity to Boyer-Moore pattern matching, that Li~\cite{Li24} included in {\tt ropebwt3}.  BML takes a threshold $L$ and searches only for MEMs of length at least $L$, generally using significantly less time than forward-backward uses to find all the MEMs.  BML can instead take a parameter $t$ and find the $t$ longest MEMs by starting with $L = 1$ and, whenever it has found at least $t$ MEMs of length $L$, resetting $L$ to 1 plus the length of the $t$th longest MEM it has found.  Both forward-backward and BML can take a frequency $f$ and find maximal substrings of the pattern that exactly match at least $f$ substrings of the text ($f$-MEMs, with regular MEMs being 1-MEMs), by stopping their searches in the FM-indexes when their BWT intervals become of length less than $f$ rather than empty (or the equivalent in the CSAs or RLCSAs if we are using those instead).

Brown et al.~\cite{BDZA+25} recently described a preprocessing step, called $k$-mer based breaking (KeBaB), to speed up MEM-finding even further.  They fix a parameter $k$ and build a Bloom filter for the distinct $k$-mers in $T$.  When given $P$, they quickly separate the $k$-mers in $P$ into those that probably occur in $T$ and those that certainly do not.  They call the maximal substrings of the pattern consisting only of the former $k$-mers {\em pseudo-MEMs}.  Since Bloom filters can return false positives but not false negatives, these pseudo-MEMs are guaranteed to contain all the MEMs of length at least $k$ of $P$ with respect to $T$.  Although they can overlap by $k - 2$ characters, their total length is often much smaller than $m$.  Brown et al.\ found that it is usually much faster to find the pseudo-MEMs and then find the MEMs in them than to find the MEMs in the pattern directly.  When $T$ is repetitive, such as a database of genomes, the Bloom filter is usually much smaller than forward-backward's or BML's indexes, so their pre-processing step also uses much less memory than the actual MEM-finding step.  We can also adapt KeBaB for finding $f$-MEMs by storing the $k$-mers in $T$ in a counting Bloom filter~\cite{FCAB00} (see also~\cite{PFDZ24}) and then, when given $P$ and $f$, finding the maximal substrings of $P$ consisting only of $k$-mers the Bloom filter says each occur at least $f$ times in $T$.

Brown et al.\ tried discarding all the pseudo-MEMs shorter than $L = 25 > k = 20$ characters when using KeBaB in a pipeline for metagenomic classification, and the pipeline became much faster and even more accurate than when using all the MEMs.  They then tried keeping only the $t = 10$ longest pseudo-MEMs and the results were even better.  Those improvements are not guaranteed, however, to carry over to other datasets or other pipelines!  The risk is that the pseudo-MEMs we keep may contain only shorter MEMs than the ones in the pseudo-MEMs we discard and that could affect the quality of the pipeline's final results.  This is because KeBaB can produce long pseudo-MEMs that do not contain long MEMs, when the Bloom filter returns false positive for a $k$-mer somewhere in the middle of the pseudo-MEMs or all the $k$-mers in the pseudo-MEM occur in $T$ but not consecutively.

This risk could be unacceptable for critical medical pipelines, for example.  To see why it is also a concern in algorithms research, consider that a longest MEM of $P$ with respect to $T$ is a longest common substring (LCS) of $P$ and $T$, and an algorithms researcher is unlikely to embrace a technique that seems to speed up LCS-finding in practice but has no real assurance of correctness.  Before we discard a pseudo-MEM, therefore, we should be confident that at least some of the ones we are keeping contain longer MEMs.  For example, if we want only the $t$ longest MEMs and we can compute a lower bound on the length of the longest MEM in each pseudo-MEM, then once we have found $t$ pseudo-MEMs guaranteed to contain MEMs of length at least $\ell$, then we can discard any pseudo-MEM of length less than $\ell$.

In this paper we show how to use parse indexing to generate pseudo-MEMs together with such lower bounds.  In Section~\ref{sec:parse_indexing} we review parse indexing, prove a key lemma about the prefix-free parses of $P$ and $T$, and redefine pseudo-MEMs in terms of those parses.  In Section~\ref{sec:new_pseudo-MEMs} we describe some useful properties of these new pseudo-MEMs, including how to compute their lower bounds.  In Section~\ref{sec:future_work} we close by discussing some directions for future work.

\section{Parse indexing}
\label{sec:parse_indexing}

Deng et al.~\cite{DHKS22} introduced parse indexing with LMS grammars but, like Hong et al.~\cite{HOKB+24}, we find it easier to consider prefix-free parsing (PFP).  To apply PFP with parameters $w$ and $p$ to a text $T [1..n]$, we choose a Karp-Rabin hash function and then run a sliding window of width $w$ over $T$ and insert a phrase break whenever the hash of the window's contents is congruent to 0 modulo $p$.  This is similar to the popular Unix tool {\tt rsync}~\cite{TM96} but after we insert each phrase break we start the next phrase with the current contents of the sliding window, so consecutive phrases overlap by $w$ characters.  This produces a parse $\PFP (T)$ that is usually much shorter than $T$, albeit over the much larger alphabet of phrases, and if $T$ is repetitive then usually $\PFP (T)$ is repetitive as well~\cite{FOGB24}.  If we are later given a pattern $P [1..m]$ and we apply PFP to it with the same parameters $w$ and $p$ and the same Karp-Rabin hash function, there is a useful relationship between $\PFP (P)$ and $\PFP (T)$:

\begin{lemma}
\label{lem:consistency}
$\PFP (P) [i..j]$ occurs at least $f$ times in $\PFP (T)$ if and only if the substring of $P$ that becomes $\PFP (P) [i..j]$ is contained in a $f$-MEM of $P$ with respect to $T$. 
\end{lemma}

\begin{proof}
If the substring of $P$ that becomes $\PFP (P) [i..j]$ is contained in a $f$-MEM of $P$ with respect to $T$, then that substring occurs at least $f$ times in $T$ and our parsing algorithm produces the sequence $\PFP (P) [i..j]$ of phrases while parsing $T$ whenever it encounters that substring.  If $\PFP (P) [i..j]$ occurs at least $f$ times in $\PFP (T)$ then the substring of $P$ that becomes $\PFP (P) [i..j]$ occurs at least $f$ times in $T$, and it follows that substring is contained in an $f$-MEM of $P$ with respect to $T$.
\end{proof}

\noindent
Lemma~\ref{lem:consistency} gives some intuition into how finding MEMs between the parses can help us find MEMs of $P$ with respect to $T$, but the details still require some care.  Specifically, for this section we redefine {\em pseudo-MEMs} to be the substrings of $P$ that become the sequences of phrases in the union of two sets:
\begin{eqnarray*}
S_1 & = & \left\{ \begin{array}{l}
	\PFP (P) [\max (i - 1, 1)..\min (j + 1, |\PFP (P)|)]\ :\\[1ex]
	\hspace{3ex} \mbox{$\PFP (P) [i..j]$ is an $f$-MEM of $\PFP (P)$ with respect to $\PFP (T)$}
	\end{array} \right\} \\[2ex]
S_2 & = & \left\{ \begin{array}{l}
	\PFP (P) [i..i + 1]\ :\\[1ex]
	\hspace{3ex} \mbox{neither $\PFP (P) [i]$ nor $\PFP (P) [i + 1]$ occur in $\PFP (T)$}
	\end{array} \right\}\,.
\end{eqnarray*}

\noindent
Our idea is to extend $f$-MEMs of $\PFP (P)$ with respect to $\PFP (T)$ by one phrase in both directions (if possible), and $S_1$ and $S_2$ are what we get when considering the $f$-MEMs that are non-empty and empty, respectively.  We try to extend by one phrase in each direction because the substring that becomes an $f$-MEM of $\PFP (P)$ with respect to $\PFP (T)$ can be contained in an $f$-MEM of $P$ with respect to $T$ that starts in the preceding phrase and ends in the following phrase, but not in one that starts more than one phrase earlier or ends more than one phrase later.

For simplicity we ignore here the uninteresting case when $\PFP (P)$ is a single phrase, but then all of $P$ is a single pseudo-MEM.  Note that if $\PFP (P)$ contains at least 2 phrases then even if a phrase is at an end of $P$, if it occurs in $T$ then it is part of a pseudo-MEM from $S_1$; if it does not then either its neighbour does, in which case the phrase is part of a pseudo-MEM from $S_1$, or it does not, in which case the phrase is part of a pseudo-MEM from $S_2$.

If we have $\PFP (T)$ and $(\PFP (T))^\rev$ indexed appropriately then, when given $P$ and $f$, we can find the $f$-MEMs of $\PFP (P)$ with respect to $\PFP (T)$ using forward-backward, which takes a number of backward steps proportional to the total number of phrases in those $f$-MEMs --- which should be much smaller than the total length of the $f$-MEMs of $P$ with respect to $T$ --- and then find the pseudo-MEMs from them.

\section{New pseudo-MEMs}
\label{sec:new_pseudo-MEMs}

Our redefinition of pseudo-MEMs is motivated by their following three properties:
\begin{enumerate}
\item The pseudo-MEMs contain all the $f$-MEMs of $P$ with respect to $T$.
\item If we delete 2 phrases from one end of any pseudo-MEM and 1 phrase from the other end, then some $f$-MEM is no longer contained in the pseudo-MEMs.
\item A pseudo-MEM cannot contain an $f$-MEM longer than itself, obviously, but it is guaranteed to contain an $f$-MEM as long as the substring that results from deleting one phrase from either end of the pseudo-MEM (although that substring is empty if the pseudo-MEM is from $S_2$, or from a pair of phrases in $S_1$ at an end of $\PFP (P)$).
\end{enumerate}
The first property also holds for the pseudo-MEMs we would choose with KeBaB when finding $f$-MEMs, but the second and third properties do not.

The first property means we can restrict our attention to the pseudo-MEMs when looking for $f$-MEMs, so they are a complete representation of $P$ for this purpose.  To see why it holds, consider an $f$-MEM $P [i..j]$ of $P$ with respect to $T$ and the longest sequence of complete phrases within it.  We consider separately the cases when that sequence is empty and when it is not.  First suppose the sequence is not empty so, by Lemma~\ref{lem:consistency}, the corresponding substring of $\PFP (P)$ occurs at least $f$ times in $\PFP (T)$ and is thus contained within an $f$-MEM of $\PFP (P)$ with respect to $\PFP (T)$.  Consider the corresponding sequence of phrases in $S_1$ and the pseudo-MEM whose parse is that sequence of phrases, which contains $P [i..j]$.

Now suppose $P [i..j]$ contains no complete phrases.  To deal with this case we should consider six subcases: when $P [i..j]$ is contained in a single phrase that occurs at least $f$ times in $\PFP (T)$; when $P [i..j]$ is contained in a single phrase that does not occur at least $f$ times in $\PFP (T)$; when $P [i..j]$ is split across two phrases, which occur together at least $f$ times in $\PFP (T)$; when $P [i..j]$ is split across two phrases, both of which occur at least $f$ times in $\PFP (T)$ but not together; when $P [i..j]$ is split across two phrases, exactly one of which occurs at least $f$ times in $\PFP (T)$; and when $P [i..j]$ is split across two phrases, neither of which occur $f$ times in $\PFP (T)$.  It is straightforward but tedious to check that in each subcase $P [i..j]$ is contained in a pseudo-MEM.  For example, when $P [i..j]$ is split across two phrases, exactly one of which occurs at least $f$ times in $\PFP (T)$, then that phrase is at the end of an $f$-MEM of $\PFP (P)$ with respect to $\PFP (T)$, and the pseudo-MEM for the corresponding sequence of phrases in $S_1$ contains $P [i..j]$.

The second property means our choice of pseudo-MEMs is somewhat parsimonious, in that we cannot shrink them very much without losing the first property.  To see why it holds, consider that we cannot delete 2 phrases from one end of a pseudo-MEM and 1 from the other end when the pseudo-MEM's parse consists of only 2 phrases, so the only pseudo-MEMs we need consider are those whose parses are in $S_1$.  Those pseudo-MEMs are obtained by adding at most one phrase to either end of an $f$-MEM of $\PFP (P)$ with respect to $\PFP (T)$, so if we delete 2 phrases from one end of one of those pseudo-MEMs and 1 from the other end then we obtain a proper substring of an $f$-MEM of $P$ with respect to $T$.

The third property gives us bounds on the longest $f$-MEM each pseudo-MEM contains.  That means that if we are interested only in the longest $t$ $f$-MEMs and we can find $t$ pseudo-MEMs guaranteed to contain $f$-MEMs of length at least $\ell$, then we can discard any pseudo-MEM of length less than $\ell$.  In other words, we can safely discard short pseudo-MEMs without worrying about accidentally discarding the longest $f$-MEMs.  To see why the third property holds, consider that if we can delete a phrase from each end of a pseudo-MEM and obtain a non-empty substring of $P$, then the pseudo-MEM's parse is in $S_1$ and so its parse occurs at least $f$ times in $\PFP (T)$ and so the substring itself occurs at least $f$ times in $T$.

\section{Future work}
\label{sec:future_work}

We realized recently that we should be able to combine KeBaB and parse indexing to obtain the advantages of both.  Suppose we build a Bloom filter for the distinct phrases in $\PFP (T)$, rather than the distinct $k$-mers, or a counting Bloom filter for the distinct phrases if we want to find $f$-MEMs.  Given $P$ and $f$, we use the Bloom filter to find the set $S_3$ containing the substrings $\PFP (P) [\max (i - 1, 1)..\min (j + 1, |\PFP (P)|)]$ such that $\PFP (P) [i..j]$ is a maximal substring of $\PFP (P)$ consisting only of phrases the Bloom filter says each occur at least $f$ times in $\PFP (T)$, and the set $S_4$ containing the substrings of $\PFP (P)$ consisting of 2 phrases neither of which the Bloom filter says occur at least $f$ times in $\PFP (T)$.  All the substrings of $\PFP (P)$ in $S_1$ are contained in substrings in $S_3$, and all the substrings in $S_2$ are contained in substrings in $S_4$, so we can use parse indexing to find $S_1$ and $S_2$ from $S_3$ and $S_4$.

We also realized that we can obtain better theoretical results if we use a parsing based on minimizers~\cite{RHHMY04} rather than PFP.  To find the minimizers in a string we choose parameters $k$ and $w > k$ and run a sliding window of length $w$ over the string, always keeping track of which $k$-mers in the window have the smallest value (either their hash values or lexicographic value); these $k$-mers are the {\em minimizers} of the string.  If we take the end of each minimizer as marking a boundary between phrases, then we obtain a parsing algorithm such that
\begin{itemize}
\item for any substring of length $\ell > 2 w - 2$ the parsing of the central $\ell - 2 w + 2$ characters is always the same;
\item no phrase is longer than $w - k + 1$.
\end{itemize}
(Our parse is not the same as the minimizer digest of the string, since we consider the entire phrases and not just the minimizers.)  The first property allows us to find pseudo-MEMs as we do with the PFP parse and the second property guarantees that --- unlike with PFP --- we cannot have very long pseudo-MEMs consisting of only a few phrases.  If we are interested only in MEMs of length greater than $2 w - 2 k + 2$, for example, then we need not consider the analogue of the $S_2$, for example.

The optimizations can easily be combined and we will investigate them both in a future paper.

\begin{credits}

\subsubsection{\ackname} This research was funded by NSERC Discovery Grant RGPIN-07185-2020.  Many thanks to Christina Boucher, Nate Brown, Eddie Ferro and Ben Langmead for helpful discussions.

\subsubsection{\discintname} The author has no competing interests to declare that are relevant to the contents of this article.

\end{credits}

% \bibliographystyle{splncs04}
% \bibliography{discarding}

\end{document}